\documentclass[a4paper, 11pt]{article}
\usepackage{amsmath}
\usepackage{amssymb,esint}
\usepackage{amscd}
\usepackage{xspace}
\usepackage{fancyhdr}
\usepackage{color}
\usepackage{srcltx}
\usepackage{slashbox}
\usepackage{makecell}
\usepackage{cite}
\usepackage{longtable}
\usepackage{float}
\usepackage{authblk}

\usepackage{enumerate}

\newenvironment{proof}[1][Proof]{\textbf{#1.} }{\hfill\rule{0.5em}{0.5em}}
{\catcode`\@=11\global\let\AddToReset=\@addtoreset
\AddToReset{equation}{section}

\newcommand{\R}{{\mathbb R}}

\newcommand{\Z}{{\mathbb Z}}
\newcommand{\N}{{\mathbb N}}
\newcommand{\F}{{\mathbb F}}

\newcommand{\C}{{\mathbb C}}

\newcommand{\cW}{{\mathcal W}}

\newcommand{\cG}{{\mathcal G}}

\newcommand{\cA}{{\mathcal A}}

\newcommand{\cL}{{\mathcal L}}
\newcommand{\cS}{{\mathcal S}}

\newcommand{\Tr}{{\rm Tr}}
\newcommand{\Ima}{{\rm Im}}

\newcommand{\Fbn}{\mathbb{F}_{2^n}}

\newcommand{\Fpn}{\mathbb{F}_{p^n}}
\newcommand{\Fpm}{\mathbb{F}_{p^m}}

\newcommand{\Fpnn}{\mathbb{F}_{p^n} \longrightarrow \mathbb{F}_{p^n}}

\newtheorem{thm}{Theorem}[section]
\newtheorem{cor}[thm]{Corollary}
\newtheorem{lem}[thm]{Lemma}

\newtheorem{definition}[thm]{Definition}
\newtheorem{example}[thm]{Example}
\newtheorem{remark}[thm]{Remark}
\newtheorem{proposition}[thm]{Proposition}

\numberwithin{equation}{section}

\begin{document}
\title{cc-differential uniformity, (almost) perfect cc-nonlinearity, and equivalences}

\author[1]{Nhan-Phu Chung}

 \author[2]{Jaeseong Jeong }
 \author[3]{Namhun Koo}
 \author[2]{Soonhak Kwon}
   \affil[1]{\footnotesize Institute of Applied Mathematics, University of Economics, Hochiminh city, Vietnam, Email: phucn@ueh.edu.vn, phuchung82@gmail.com. }
   \affil[2]{\footnotesize Applied Algebra and Optimization Research Center, Sungkyunkwan University, Suwon, Korea, Email: wotjd012321@naver.com, shkwon@skku.edu.}
   \affil[3]{\footnotesize Institute of Mathematical Sciences, Ewha Womans University, Seoul, Korea, Email: nhkoo@ewha.ac.kr.}
\maketitle

\begin{abstract}
In this article, we introduce new notions $cc$-differential uniformity, $cc$-differential spectrum, PccN functions and APccN functions, and investigate their properties. We also introduce $c$-CCZ equivalence, $c$-EA equivalence, and $c1$-equivalence. We show that $c$-differential uniformity is invariant under $c1$-equivalence, and $cc$-differential uniformity and $cc$-differential spectrum are preserved under $c$-CCZ equivalence. We characterize $cc$-differential uniformity of vectorial Boolean functions in terms of the Walsh transformation. We investigate $cc$-differential uniformity of power functions $F(x)=x^d$. We also illustrate examples to prove that $c$-CCZ equivalence is strictly more general than $c$-EA equivalence.
\end{abstract}
\section{Introduction} Let $n,s\in \N$ and $p$ be a prime number. We denote by $\F_{p^n}$ the finite field of $p^n$ elements and $\F_{p^n}^\times$ the subset consisting of all invertible elements in $\F_{p^n}$. In \cite{Nyberg}, to measure the resistance of the block cipher against the
differential cryptanalysis, Nyberg introduced the difference distribution table $\Delta_F(a,b)$ and the differential uniformity $\Delta_F$ for a vectorial function $F:\F_{p^n}\to \F_{p^s}$, which is used as an $S$-box inside a block cipher, as follows
\begin{align*}
\Delta_F(a,b):&=\#\{x\in \F_{p^n}: F(x+a)-F(x)=b\} \mbox{ for every } a\in \F_{p^n}, b\in \F_{p^s},\\
\Delta_F:&=\max\{\Delta_F(a,b): a\in \F_{p^n}^\times, b\in \F_{p^s}\}.
\end{align*}
Then the lower value of $\Delta_F$ is, the better $F$ resists a differential attack. We say that $F$ is perfect nonlinear (PN) if $\Delta_F=1$, and almost perfect nonlinear (APN) if $\Delta_F=2$. After the work of Nyberg, differential uniformity, PN functions, and APN functions have been studied intensively by numerous authors, see for details and references in a recent monograph \cite{Carlet}.

Recently, given $c\in \F_{p^s}^\times$, Ellingsen, Felke, Riera, Stănică and Tkachenko \cite{EFRST} introduced the $c$-difference distribution
table and the $c$-differential uniformity for functions $F:\F_{p^n}\to \F_{p^s}$ coinciding with the difference distribution table and the differential uniformity, respectively, when $c=1$. This new topics attracts people in the field and hence many results have been established \cite{BKM,BT, HPR+21,HPS, JKK22, LRS, MMM, WZ,WLZ, Y,ZH}. Let us review the definitions of the $c$-difference distribution
table and the $c$-differential uniformity.

 Let $F:\F_{p^n}\to \F_{p^s}$ be a function. We define the \textit{$c$-derivative} of $F$ with respect to $a\in \F_{p^n}$ by
$${_{c}}D_{a}F(x)=F(x+a)-cF(x) \mbox{ for all } x\in \F_{p^n}.$$
Given $a\in \F_{p^n}$, $b\in \F_{p^s}$, we let ${_{c}}\Delta_{F}(a,b)=\#\{x\in \F_{p^n}: F(x+a)-cF(x)=b\}$. We call ${_{c}}\Delta_{F}=\max\{{_{c}}\Delta_{F}(a,b): a\in \F_{p^n}, b\in \F_{p^s} \mbox{ and } a\neq 0 \mbox{ if } c= 1\}$ the \textit{$c$-differential uniformity} of $F$. The \textit{$c$-differential spectrum} of $F$ is defined as
${_{c}}D_{F}:=\{{_{c}}\Delta_{F}(a,b): a\in \F_{p^n}, b\in \F_{p^s}\}$.
We say that $F$ is a \textit{perfect c-nonlinear} (PcN) function if ${_{c}}\Delta_{F}=1$. If ${_{c}}\Delta_{F}=2$ then $F$ is called an \textit{almost perfect c-nonlinear} (APcN) function.

Inspired by the work \cite{EFRST}, we introduce new notions $cc$-differential uniformity, PccN functions and APccN functions as follows.
 Let $F:\F_{p^n}\to \F_{p^s}$ be a function and $c\in \F_{p^{\gcd(n,s)}}^\times$. The \textit{$cc$-derivative} of $F$ with respect to $a\in \F_{p^n}$ is the function ${_{cc}}D_{a}F:\F_{p^n}\to \F_{p^s}$ defined by
$${_{cc}}D_{a}F(x)=F(cx+a)-cF(x) \mbox{ for all } x\in \F_{p^n}.$$
We define the \textit{$cc$-difference distribution table} of $F$ by
$${_{cc}}\Delta_{F}(a,b)=\#\{x\in \F_{p^n}: F(cx+a)-cF(x)=b\},$$ for every $a\in \F_{p^n}$, $b\in \F_{p^s}$. 
The \textit{cc-differential uniformity} and   \textit{$cc$-differential spectrum} of $F$ are defined as
$${_{cc}}\Delta_{F}=\max\{{_{cc}}\Delta_{F}(a,b): a\in \F_{p^n}, b\in \F_{p^s} \mbox{ and } a\neq 0 \mbox{ if } c= 1\},$$
and $${_{cc}}D_{F}:=\{{_{cc}}\Delta_{F}(a,b): a\in \F_{p^n}, b\in \F_{p^s}\}, \mbox{ respectively.}$$
We say that $F$ is a \textit{perfect cc-nonlinenar} (PccN) function if ${_{cc}}\Delta_{F}=1$, and an \textit{almost perfect cc-nonlinear} (APccN) function if ${_{cc}}\Delta_{F}=2.$

Now we will explain that our $cc$-differential uniformity is a tool to measure the resistance of a cipher against a specific differential attack on the operation $\circ_{cc}$ on $\F_{p^n}$ and $\F_{p^s}$.

We define the operation $\circ_{cc}$ on $\F_{p^n}$ and $\F_{p^s}$ by
\begin{align*}
x\circ_{cc} y:= cx+y, \mbox{ for every } x\in \F_{p^n} \mbox{ or } x\in \F_{p^s}.
\end{align*}
For every $a\in \F_{p^n}$ and $b\in \F_{p^s}$ we define
\begin{align*}
{_{\circ_{cc}}}\Delta_F(a,b):=\#\{x\in \F_{p^n}: F(x\circ_{cc} a)=F(x)\circ_{cc}b\}
\end{align*}
and ${_{\circ_{cc}}}\Delta_F:=\max\{{_{\circ_{cc}}}\Delta_F(a,b): a\in C, b\in \F_{p^s}\}$, where $C$ is the set of all $a\in \F_{p^n}$ such that there exists $x\in \F_{p^n}$ with $x\circ_{cc}a\neq x$.

For every $a\in \F_{p^n}$ and $b\in \F_{p^s}$ we have that
\begin{align*}
F(cx+a)-cF(x)=b\Leftrightarrow F(cx+a)=cF(x)+b\Leftrightarrow F(x\circ_{cc} a)=F(x)\circ_{cc} b.
\end{align*}
Therefore, ${_{cc}}\Delta_{F}(a,b)={_{\circ_{cc}}}\Delta_F(a,b)$ for every $a\in \F_{p^n}$ and $b\in \F_{p^s}$. On the other hand, $x\circ_{cc}a=x$ for all $x\in \F_{p^n}$ if and only if $a=0$ and $c=1$. Hence ${_{cc}}\Delta_{F}={_{\circ_{cc}}}\Delta_{F}$.
\begin{remark}
\label{R-cc-differentials}
\begin{enumerate}
\item For the case $c=1$, the equation $F(cx+a)-cF(x)=b$ always has no solutions for $a=0$ and $b\neq 0$. Therefore in this case, we have that $${_{cc}}\Delta_{F}=\max\{{_{cc}}\Delta_{F}(a,b): (a,b)\in \F_{p^n}\times \F_{p^s}\setminus \{(0,0)\}\}.$$
\item In \cite{BKM}, for a fixed $c\in \F_{p^n}$ the authors introduced the operation $\circ_c$ on $\F_{p^n}$ by $x\circ_c y:=x+cy$ for every $x,y\in \F_{p^n}$. Given a function $F:\F_{p^n}\to \F_{p^n}$, they also defined
\begin{align*}
{_{\circ_{c}}}\Delta_F(a,b):=\#\{x\in \F_{p^n}: F(x\circ_{c} a)=b\circ_{c}F(x)\},
\end{align*}
and ${_{\circ_{c}}}\Delta_F:=\max\{{_{\circ_{c}}}\Delta_F(a,b): a\in D, b\in \F_{p^n}\}$, where $D$ is the set of all $a\in \F_{p^n}$ such that there exists $x\in \F_{p^n}$ with $x\circ_{c}a\neq x$.

For every $a,b,x\in \F_{p^n}$, we have
\begin{align*}
F(x+ca)-cF(x)=b\Leftrightarrow F(x+ca)=b+cF(x)\Leftrightarrow F(x\circ_{c} a)=b\circ_{c} F(x).
\end{align*}
Hence ${_{c}}\Delta_{F}(a,b)={_{\circ_{c}}}\Delta_F(a,b)$ for every $a,b\in \F_{p^n}$ and therefore ${_{c}}\Delta_{F}={_{\circ_{c}}}\Delta_F$.
\end{enumerate}
\end{remark}

On the other hand, there are two important notions of equivalence among vectorial Boolean functions. The first one is Carlet–Charpin–Zinoviev (CCZ)
equivalence. This notion was introduced in \cite{CCZ} and the term CCZ equivalence was used from \cite{BCP}. The second one is extended affine equivalence (EA-equivalence), which is a special case of CCZ-equivalence. It is known that differential uniformity is invariant under CCZ-equivalence and EA-equivalence. However, $c$-differential uniformity is not invariant under CCZ-equivalence \cite{BKM, HPR+21, JKK22}. In Section 2, we introduce $c$-CCZ equivalence, $c$-EA equivalence, and $c1$-equivalence. We show that $c$-differential uniformity is invariant under $c1$-equivalence, and $cc$-differential uniformity and $cc$-differential spectrum are preserved under $c$-CCZ equivalence. We also illustrate examples to prove that $c$-CCZ equivalence is strictly more general than $c$-EA equivalence.

In Section 3, using a method of Carlet \cite{Carlet18} generalizing a result of Chabaud and Vaudenay \cite{CV}, we characterize $cc$-differential uniformity of vectorial Boolean functions in terms of the Walsh transformation. In Section 4, we present several properties of $cc$-differential uniformity, differences between $c$-differential uniformity and $cc$-differential uniformity. We also investigate $cc$-differential uniformity of power functions $F(x)=x^d$ on $\F_{p^n}$.

To wrap up the introduction, we recall definitions of trace functions. Let $n\in \N$ and $s$ be a divisor of $n$. We denote $\Tr^s_n(x)$ the trace function from $\F_{p^n}$ to $\F_{p^s}$, that is
$$\Tr^s_n(x):=\sum_{j=0}^{n/s-1}x^{p^{js}} \mbox{ for every } x\in \F_{p^n}.$$
When $s=1$ we write $\Tr_n(x)$ instead of $\Tr_n^s(x)$.

\textbf{Acknowledgement:} N-P. Chung is funded by University of Economics Ho Chi Minh City, Vietnam.
\section{$c$-equivalences}
In this section we will introduce $c1$-equivalence, $c$-CCZ equivalence and $c$-EA equivalence. Given $n,s\in \N$ and $c\in \F_{p^{\gcd(n,s)}}^\times$.

\begin{definition}
Let  $A : \F_{p^n}\to \F_{p^s}$ be an affine function with linear part
$L(x)\overset{\rm{def}}{=}A(x)-A(0)$.
We say $A(x)$ is a $c$-affine function if $A(cx)=cA(x)-(c-1)A(0)$
(i.e., $L(cx)=cL(x)$) for all $x\in \Fpn$.
\end{definition}

\begin{definition}
Two functions $F, F': \F_{p^n}\to \F_{p^s}$ are said to
be
 \begin{enumerate}
  \item[--] $c1$-equivalent if $F' = A_1 \circ F \circ A_2$,
 where $A_1:\F_{p^s}\to \F_{p^s}$ is a  $c$-affine permutation and $A_2: \F_{p^n}\to \F_{p^n}$ is an $1$-affine permutation on $\Fpn$.
 \item[--] $c$-affine equivalent if $F' = A_1 \circ F \circ A_2$,
 where $A_1: \F_{p^s}\to \F_{p^s}$ and $A_2:\F_{p^n}\to \F_{p^n}$ are $c$-affine permutations.
 \item[--] $c$-extended affine equivalent ($c$-EA equivalent) if
 $F' = A_1 \circ F \circ A_2  + A_3$, where $A_1:\F_{p^s}\to \F_{p^s}$, $A_2:\F_{p^n}\to \F_{p^n}$ are $c$-affine
 permutations and $A_3: \F_{p^n}\to \F_{p^s}$ is a $c$-affine function.
 \item[--] $c$-CCZ equivalent if there
exists a $c$-affine permutation $\mathcal{A}$ on $\Fpn \times \F_{p^s}$
such that $\mathcal{A}(\cG_F)= \cG_{F'}$ where $\cG_F$ is the graph of
the function $F$, that is $\cG_F = \{(x,F(x)) :  x \in \Fpn\}
\subset \Fpn \times \F_{p^s}$.
 \end{enumerate}
\end{definition}

\begin{remark}
\begin{enumerate}
\item It is not difficult to check that if $A:\F_{p^n}\to \F_{p^n}$ is a $c$-affine permutation then so its inverse $A^{-1}$ is. In addition, the composition of two $c$-affine permutations is also a $c$-affine permutation. Hence $c1$-equivalence, $c$-affine equivalence, $c$-EA equivalence, and $c$-CCZ equivalence are equivalence relations on the set of all maps from $\F_{p^n}$ to $\F_{p^s}$.
\item Since a $c$-affine function is trivially an $1$-affine function (i.e., affine function),
 $c$-affine equivalence implies $c1$-equivalence.
 Note that $c1$-equivalence is a generalization of so called A-equivalence in \cite[page 236]{HPR+21}
where the special case $A_1(x)={\rm id}, A_2(x)={\rm affine\,\,
function}$ was considered. 
\item It is clear that $F$ and $F'$ are $c$-EA equivalence if and only if $F' = L_1 \circ F \circ A_2  + A_3$, where $L_1: \F_{p^s}\to \F_{p^s}$ is a $c$-linear
 permutation, $A_2:\F_{p^n}\to \F_{p^n}$ is a $c$-affine permutation, and $A_3:\F_{p^n}\to \F_{p^s}$ is a $c$-affine function.
\item If $c\in \F_{p}^\times$ then our definitions of $c$-CCZ equivalence and $c$-EA equivalence coincide with the corresponding usual CCZ equivalence and EA equivalence, respectively.
\end{enumerate}
\end{remark}

Now we will show that $c$-EA equivalence is a special case of $c$-CCZ equivalence.
\begin{lem}
If two functions $F,F':\F_{p^n}\to \F_{p^s}$ are $c$-EA equivalent, then they are
$c$-CCZ equivalent.
\end{lem}
\begin{proof}
First we will prove that if $F$ and $F'$ are $c$-affine equivalent then it is $c$-CCZ equivalent. Let $A_1:\F_{p^n}\to \F_{p^n}$, $A_2:\F_{p^s}\to \F_{p^s}$ be $c$-affine permutations such that $F'=A_2\circ F\circ A_1$. We define the map $\cA:\F_{p^n}\times \F_{p^s}\to \F_{p^n}\times \F_{p^s}$ by $A(x,y):=(A_1^{-1}(x),A_2(y))$ for every $(x,y)\in \F_{p^n}\times \F_{p^s}$. Let $L_1:\F_{p^n}\to \F_{p^n}$ be the linear part of $A_1$, i.e. $L_1(x)=A_1(x)-A_1(0)$ for every $x\in \F_{p^n}$. Then for every $y=L_1(x)\in \F_{p^s}$ we have $cy=cL_1(x)=L_1(cx)$ and hence $L^{-1}_1(cy)=cx=cL_1^{-1}(y)$. Thus, $A_1^{-1}$ is a $c$-affine and therefore so is $\cA$.
On the other hand, for every $x\in \F_{p^n}$ we have $\cA(x,F(x))=((A_1^{-1}(x),A_2(F(x)))=(y,A_2\circ F\circ A_1(y))$ with $y=A_1^{-1}(x)$. Hence $F$ and $F'$ are $c$-CCZ equivalent.

Next, if $A_3:\F_{p^n}\to \F_{p^n}$ is a $c$-affine map such that $F'(x)=F(x)+A_3(x)$ for every $x\in \F_{p^n}$ then the map $A(x,y)=(x,y+A_3(x))$ is a $c$-affine permutation on $\F_{p^n}\times \F_{p^s}$ mapping $\cG_F$ onto $\cG_{F'}$. Therefore, $F$ and $F'$ are $c$-CCZ equivalent.
\end{proof}

Next we will show that $cc$-differential spectrum and $cc$-differential uniformity are preserved under $c$-CCZ equivalence.

\begin{thm}
Let $n,s\in  \N$ and $c\in \F_{p^{\gcd(n,s)}}^\times$. Let $F, H:\F_{p^n}\to \F_{p^s}$ be two functions that are $c$-CCZ equivalent. Then ${_{cc}}D_{F}= {_{cc}}D_{H}$ and ${_{cc}}\Delta_{F}={_{cc}}\Delta_{H}$.
\end{thm}
\begin{proof}
We denote $c\cG_F:=\{(cx,cF(x)):x\in \F_{p^n}\}$. Let $\varphi:\F_{p^n}\times \F_{p^s}\to \F_{p^n}\times \F_{p^s}$ be a bijective linear map and $X_0=(u,v)\in \F_{p^n}\times \F_{p^s}$ such that $\cG_F=\varphi(\cG_H)+X_0$ and $\varphi(cX)=c\varphi(X)$ for every $X\in \F_{p^n}\times \F_{p^s}$.
Let $(a,b)\in \F_{p^n}\times \F_{p^s}$. Then ${_{cc}}\Delta_{F}(a,b)$ is
the number of solutions $(X,Y)\in \cG_F\times (c\cG_F)$ of the equations $X-Y=(a,b)$.
Assume that $(X,Y)\in \cG_F\times (c\cG_F)$ is a solution of the equation $X-Y=(a,b)$. Let $(X_1,Y_1)\in \cG_H\times \cG_H$ be the unique solution of the equations
$X=\varphi(X_1)+X_0$, $Y=c(\varphi(Y_1)+X_0)$. Then $\varphi(X_1)+X_0-c\varphi(Y_1)-cX_0=(a,b)$ and hence $\varphi(X_1-cY_1)=(a,b)+(c-1)X_0$. Therefore $(X,Y)\in \cG_F\times (c\cG_F)$ is a solution of $X-Y=(a,b)$ if and only if $(X_1,cY_1)\in \cG_H\times (c\cG_H)$ is a solution of the equation $X_1-cY_1=\varphi^{-1}((a,b)+(c-1)(u,v))=(a_1,b_1)$. As ${_{cc}}\Delta_{H}(a_1,b_1)$ is the number of solutions $(X_1,cY_1)\in \cG_H\times (c\cG_H)$ of the equation $X_1-cY_1=(a_1,b_1)$, we get that ${_{cc}}\Delta_{F}(a,b)={_{cc}}\Delta_{H}(a_1,b_1)$. As $\varphi$ is bijective we get that ${_{cc}}D_{F}= {_{cc}}D_{H}$. For the case $c\neq 1$, as ${_{cc}}\Delta_{F}=\max\{{_{cc}}\Delta_{F}(a,b): a\in \F_{p^n}, b\in \F_{p^s}\}$ we get that ${_{cc}}\Delta_{F}={_{cc}}\Delta_{H}$. For the case $c=1$, as ${_{cc}}\Delta_{F}=\max\{{_{cc}}\Delta_{F}(a,b): a\in \F_{p^n}, b\in \F_{p^s}\mbox{ and } a\neq 0\}$, and $\varphi^{-1}((a,b))=(0,0)$ if and only if $(a,b)=(0,0)$, combining with Remark \ref{R-cc-differentials}, we get also ${_{cc}}\Delta_{F}={_{cc}}\Delta_{H}$ for this case.
\end{proof}

The following result is our version of \cite[Proposition 3]{BCP} for the $c$-EA equivalence. We use the same techniques there to prove it and we present its proof here for the completeness.
\begin{lem}
Let $F,G:\F_{p^n}\to \F_{p^n}$ and $c\in \F_{p^n}^\times$. Then we have
\begin{enumerate}
\item Then $G$ is $c$-EA equivalent to $F^{-1}$ (if $F$ is a permutation) if and only if there exist $(u,v)\in \F_{p^n}\times\F_{p^n}$ and a $c$-linear permutation $\cL=(L_1,L_2):  \F_{p^n}\times\F_{p^n}\to  \F_{p^n}\times\F_{p^n}$ such that $L_2$ depends only on $y$, i.e. $L_1(x,y)=L(y)$ for every $x,y\in \F_{p^n}$ and $\cA(\cG_F)=\cG_{G}$, where $\cA(x,y)=\cL(x,y)+(u,v)$ for every $(x,y)\in  \F_{p^n}\times\F_{p^n}$.
\item $G$ is $c$-EA equivalent to $F$ if and only if there exist $(u,v)\in \F_{p^n}\times\F_{p^n}$ and a $c$-linear permutation $\cL=(L_1,L_2):  \F_{p^n}\times\F_{p^n}\to  \F_{p^n}\times\F_{p^n}$ such that $L_1$ depends only on $x$, i.e. $L_1(x,y)=L(x)$ for every $x,y\in \F_{p^n}$ and $\cA(\cG_F)=\cG_{G}$, where $\cA(x,y)=\cL(x,y)+(u,v)$ for every $(x,y)\in  \F_{p^n}\times\F_{p^n}$.
\end{enumerate}
\end{lem}
\begin{proof}
1) ``$\Rightarrow$'' Let $G(x)=R_1\circ F^{-1}\circ (R_2(x)+a)+R_3(x)$ for every $x$, where $R_1, R_2:\F_{p^n}\to \F_{p^n}$ are $c$-linear permutations, $a\in \F_{p^n }$, and $R_3$ is a $c$-affine map on $\F_{p^n}$. Take $\cL(x,y)=(L_1(x,y),L_2(x,y))=(R_2^{-1}(y),R_3\circ R_2^{-1}(y)+R_1(x))$ for every $x,y\in \F_{p^n}$. Then $\cL$ is a $c$-linear permutation and $L_1(x,y)= R_2^{-1}(y)$. On the other hand, put $u=-R_2^{-1}(a),v=-R_3\circ R_2^{-1}(a)$, we have
\begin{align*}
\cL(x,F(x))+(u,v)&=(R_2^{-1}\circ F(x)-R_2^{-1}(a),R_3\circ R_2^{-1}\circ F(x)+R_1(x)-R_3\circ R_2^{-1}(a))\\
&=(z,R_1\circ F^{-1}(R_2(z)+a)+R_3(z))\\
&=(z,G(z)),
\end{align*}
with $z=R_2^{-1}\circ F(x)-R_2^{-1}(a)$.

Now we will prove the converse ``$\Leftarrow$''. We consider functions $R_1, R_2: \F_{p^n}\to \F_{p^n}$ defined by $R_1(x)=L_2(x,0)$ and $R_2(y)=L_2(0,y)$ then $R_1$ and $R_2$ are linear, $R_1(cx)=cR_1(x), R_2(cy)=cR_2(y)$ and $L_2(x,y)=R_1(x)+R_2(y)$ for every $x,y\in \F_{p^n}$. Hence $F_1(x)=L_1(x,F(x))=L\circ F(x)$, $F_2(x)=L_2(x, F(x))=R_1(x)+R_2\circ F(x)$. As $\cA(\cG_F)$ is the graph of the function $G$, we must have that $F_1$ is a permutation. Hence both $L$ and $F$ are permutations. If there exists $x\in \F_{p^n}\setminus \{0\}$ such that $R_1(x)=0$ then $(x,0)\neq (0,0)$ is a solution of the systems $L_1(x,y)=0$, $R_1(x)+R_2(y)=0$ contradicting with the permutation property of $\cL$. Hence $x=0$ is the only solution of $R_1(x)=0$ and therefore $R_1$ is a permutation. For every $y\in \F_{p^n}$, there exists a unique $x\in \F_{p^n}$ such that $(y,G(y))=(F_1(x)+u,F_2(x)+v)$. Hence
\begin{align*}
G(y)&=F_2(F_1^{-1}(y-u))+v\\
&=R_1\circ F^{-1}\circ L^{-1}(y-u)+R_2\circ F \circ F^{-1}\circ L^{-1}(y-u)+v\\
&=R_1\circ F^{-1}\circ B_1(y)+B_2(y),
\end{align*}
where $B_1(y)=L^{-1}(y)-L^{-1}(u)$ and $B_2(y)=R_2\circ L^{-1}(y)-R_2\circ L^{-1}(u)+v$ are $c$-affine permutations over $\F_{p^n}$. Hence $G$ is $c$-EA equivalent to $F^{-1}$.

The proof of 2) is also similar.
\end{proof}
\begin{lem}
Let $F$ be a permutation on $\F_{p^n}$. Then for every $c\in \F_{p^n}^\times$ we have ${_{cc}}\Delta_F={_{cc}}\Delta_{F^{-1}}$. In particular, if $F$ is PccN (APccN) then so is $F^{-1}$. Furthermore, $F$ and $F^{-1}$ are $c$-CCZ equivalent.
\end{lem}
\begin{proof}
Let $a,b\in \F_{p^n}$. Then $x$ is a solution of the equation $F(cx+a)-cF(x)=b$ if and only if $y=F(x)$ is a solution of the equation $F^{-1}(cy+b)-cF^{-1}(y)=a$.

Let $T:\F_{p^n}\to \F_{p^n}$ defined by $T(x,y)=(y,x)$ for every $x,y\in \F_{p^n}$. Then $T$ is a $c$-linear permutation and $T(\cG_F)=\cG_{F^{-1}}$.
\end{proof}

\begin{lem}\label{ablinearmap lem}
    Let $\cL$ be the linear permutation of $\Fpn \times \Fpm$ with $\cL(x,F(x)) = (F_1(x),F_2(x))$.
    For $a \in \Fpn^\times$, $b \in \Fpm$, the linear permutation $\cL_{a,b}(x,F(x)) = (aF_1(x),bF_2(x))$  maps the graph of $F$ to the graph of
    $x \mapsto bF_2\left (F_1^{-1}\left(\frac{x}{a}\right) \right ) .$
\end{lem}
\begin{proof}
    Let $a \in \Fpn^\times$, $b \in \Fpm$. It is clear that the inverse map of $x \mapsto aF_1(x)$ is $x \mapsto F_1^{-1}(\frac{x}{a})$. Then it holds that
    $$ \left \{(aF_1(x),bF_2(x)) : x \in \Fpn \right \} =\left  \{\left (x, bF_2\left (F_1^{-1}\left(\frac{x}{a}\right) \right )\right ) : x \in \Fpn\right  \},$$
    which completes the proof.
\end{proof}

It is known that CCZ equivalence is strictly more general EA equivalence for $p=2$ \cite[Theorem 3]{BC} and $p>2$ \cite[Proposition 7]{BH}.
Next we will illustrate examples to show that $c$-CCZ equivalence is strictly more general than $c$-EA equivalence.

Let us recall the algebraic degree of a function. Every function $F:\F_{p^n}\to \F_{p^n}$ is uniquely written as a univariate polynomial of degree smaller than $p^n$ as follows
$$F(x)=\sum_{j=0}^{p^n-1}a_jx^j, \mbox{ } a_j\in \F_{p^n}.$$
Let $m$ be an integer in $[0,p^n)$. Then $m=\sum_{j=0}^{n-1}m_jp^j$ with $0\leq m_j<p$ for every $j$. We define $w_p(m):=\sum_{j=0}^{n-1}m_j$. The algebraic degree $d^\circ (F)$ of $F$ is defined by
$$d^\circ (F):=\max\{w_p(m):0\leq m<p^n, a_m\neq 0\}.$$
Let $s$ be a divisor of $n$ then every function $F:\F_{p^n}\to \F_{p^s}$ can be considered as a function from $\F_{p^n}$ to itself, and hence we also can define its algebraic degree.

The two following propositions show that $c$-CCZ equivalence is strictly more general than $c$-EA equivalence for every prime number $p$.

\begin{proposition} Let $m\geq 4$ even, $F:\F_{2^m}\to \F_{2^m}$, $F(x)=x^{2^i+1}$, $\gcd(m,i)=1$. Then for every $c \in \F_{2^m}^\times$ we have
    $F(x)=x^{2^i+1}$ and $F''(x) =\frac{x^{2^i+1}}{c^{2^i}} + \left (\frac{x^{2^i}}{c^{2^i-1}} + x + c\right )  \Tr_m\left ( \frac{x^{2^i+1}}{c^{2^i+1}}\right )$ are $c$-CCZ equivalent but $c$-EA inequivalent.
\end{proposition}
\begin{proof} We define the map $\cL:\F_{2^m}\to \F_{2^m}$ by $\cL(x,y)=(cx+c\Tr_m(y) , cy)$ for every $x,y\in \F_{2^m}$. Then $\cL$ is a $c$-linear permutation. We define two maps $F_1,F_2:\F_{2^m}\to \F_{2^m}$ by $F_1(x) = x+\Tr_m(F(x))$ and $F_2(x) = F(x)$ for every $x\in \F_{2^m}$. Note that $F_1$ is an involution, since
    \begin{equation*}
        \begin{split}
            F_1(F_1(x)) = F_1(x+\Tr_m(x^{2^i+1})) &= x + \Tr_m(x^{2^i+1}) + \Tr_m(x^{2^i+1} + (x^{2^i} + x+ 1) \Tr_m(x^{2^i+1})) \\
            &= x + \Tr_m( (x^{2^i} + x+ 1) \Tr_m(x^{2^i+1})) \\
            &= x + \Tr_m(x^{2^i+1}) \Tr_m( x^{2^i} + x+ 1 ) \\
            &= x + \Tr_m(x^{2^i+1}) \Tr_m( 1)  = x \quad (\because m
            \text{ is even})
        \end{split}
    \end{equation*}
    Then, for every $x\in \F_{2^m}$ we have $L(x,F(x))=(cF_1(x),cF_2(x))=(y, cF_2( F_1^{-1}(\dfrac{y}{c}))$, where $y=cF_1(x)$.
    Hence the $c$-linear permutation $\cL_{c,c}$ maps the graph of $F$ to the graph of $F''(x) = cF(F_1^{-1}(\frac{x}{c}))$. 
    On the other hand, the function $F''$ can be written as
    \begin{equation*}
        \begin{split}
            F''(x) =  cF\left (F_1^{-1}\left (\dfrac{x}{c}\right )\right ) &=  cF\left (F_1\left (\dfrac{x}{c}\right )\right )\\
            &= c\left (\dfrac{x}{c} + \Tr_m\left ( \dfrac{x^{2^i+1}}{c^{2^i+1}}\right )\right )^{2^i+1} \\
            &= \dfrac{x^{2^i+1}}{c^{2^i}} + \left (\dfrac{x^{2^i}}{c^{2^i-1}} + x + c\right )   \Tr_m\left ( \dfrac{x^{2^i+1}}{c^{2^i+1}}\right ).
        \end{split}
    \end{equation*}
    Therefore $F$ and $F''$ are $c$-CCZ equivalent. It is easy to check $F$ has algebraic degree $2$ and $F''$ has algebraic degree $3$. Hence $F$ and $F''$ are $c$-EA inequivalent.
\end{proof}
\begin{proposition} Let $p$ be an odd prime, $n \geq 3$, and $m > 1$ be a divisor of $n$. Then for every $c \in \Fpm^\times$ we have
        $F(x) = \Tr^m_n(x^2-x^{p+1})$ and $F''(x) = c\Tr^m_n(\frac{x^2}{c^2}-\frac{x^{p+1}}{c^{p+1}})-c\Tr_m(\frac{x^2}{c^2}-\frac{x^{p+1}}{c^{p+1}})\Tr^m_n(\frac{x^p}{c^p}-\frac{x}{c})$ are $c$-CCZ equivalent but $c$-EA inequivalent.
\end{proposition}
\begin{proof} We define the map $\cL:\F_{2^m}\to \F_{2^m}$ by $\cL(x,y)=(x+\Tr_m(y) , y)$ for every $x,y\in \F_{2^m}$. Then $\cL$ is a $c$-linear permutation. We define two maps $F_1, F_2: \F_{p^m}\to \F_{p^m}$ by $F_1(x) = x+\Tr_m(F(x))$ and $F_2(x) = F(x)$ for every $x\in \F_{p^m}$.
From the proof of \cite[Proposition 7]{BH}, we know that $F_1$ is bijective, $F_1^{-1}(x) = x- \Tr_n( x^2 -x^{p+1})$, and $\cL$ maps the graph of $F$ to the graph of $$F'(x) = \Tr^m_n(x^2-x^{p+1})+\Tr_n(x^2-x^{p+1})\Tr^m_n(x^p-x).$$

Then applying Lemma \ref{ablinearmap lem} we get that the induced linear permutation $\cL_{c,c}(x,y) = (cx+c\Tr_m(y) , cy)$ maps the graph of $F$ to the graph of $F''(x) = cF(F_1^{-1}(\frac{x}{c}))$. On the other hand, the function $F''$ can be written as
\begin{equation*}
    \begin{split}
        F''(x) &= cF\left (F_1^{-1}\left (\frac{x}{c}\right )\right ) = cF\left (\frac{x}{c} -\Tr_n\left (\frac{x^2}{c^2}-\frac{x^{p+1}}{c^{p+1}}\right )\right ) \\
        &= c\Tr^m_n\left (\left(\frac{x}{c} -\Tr_n\left (\frac{x^2}{c^2}-\frac{x^{p+1}}{c^{p+1}}\right )\right) ^2-\left(\frac{x}{c} -\Tr_n\left (\frac{x^2}{c^2}-\frac{x^{p+1}}{c^{p+1}}\right )\right) ^{p+1}\right ) \\
        &=c \Tr^m_n\left(  \frac{x^2}{c^2} -
        \frac{2x}{c} \Tr_n\left (\frac{x^2}{c^2}-\frac{x^{p+1}}{c^{p+1}}\right ) +\Tr_n\left (\frac{x^2}{c^2}-\frac{x^{p+1}}{c^{p+1}}\right )^2 \right. \\
        & \hspace{5em}\left.     -\frac{x^{p+1}}{c^{p+1}}+ \left(\frac{x^p}{c^p} + \frac{x}{c}\right)\Tr_n\left (\frac{x^2}{c^2}-\frac{x^{p+1}}{c^{p+1}}\right ) - \Tr_n\left (\frac{x^2}{c^2}-\frac{x^{p+1}}{c^{p+1}}\right )^2\right) \\
        &= c\Tr^m_n\left(  \frac{x^2}{c^2}-\frac{x^{p+1}}{c^{p+1}}  + \left(\frac{x^p}{c^p}  - \frac{x}{c}\right)\Tr_n\left (\frac{x^2}{c^2}-\frac{x^{p+1}}{c^{p+1}}\right )\right) \\
        &= c\Tr^m_n\left(  \frac{x^2}{c^2}-\frac{x^{p+1}}{c^{p+1}}\right)  + c\Tr_m\left (\frac{x^2}{c^2}-\frac{x^{p+1}}{c^{p+1}}\right )\Tr^m_n\left( \frac{x^p}{c^p}  - \frac{x}{c}\right).
    \end{split}
\end{equation*}
Therefore $F$ and $F''$ are $c$-CCZ-equivalent. It is easy to check $F$ has algebraic degree $2$ and $F''$ has algebraic degree $3$. Hence $F$ and $F''$ are $c$-EA-inequivalent.
\end{proof}

Although in general, $c$-CCZ equivalence does not coincide with $c$-EA equivalence, they are indeed the same when we restrict to some classes of functions. Using the same proof as in \cite[Theorem 3]{BH1}, we get the following lemma.
\begin{lem}
Let $p$ be a prime number and $n,s\in \N$. We denote by $\cS$ the set of all function $F$ from $\F_{p^n}\to \F_{p^s}$ such that all its derivatives $D_aF(x)=F(x+a)-F(x)$ are surjective for every $a\in \F_{p^n}^\times$. Then for every $c\in \F_{p^{\gcd(n,s)}}^\times$, $c$-CCZ equivalence coincides with $c$-EA equivalence on $\cS$.
\end{lem}
\begin{cor}
Let $p$ be a prime number, $n\in \N$ and $c\in \F_{p^n}^\times$. Let $F,G:\F_{p^n}\to \F_{p^n}$ such that $F$ and $G$ are $c$-CCZ equivalent and $F$ is a PN function. Then $F$ and $G$ are $c$-EA equivalent.
\end{cor}

Next we will present examples being CCZ-equivalent but not $c$-CCZ equivalent.
\begin{example}
(i) Let $m\geq 4$ even, $F:\F_{2^m}\to \F_{2^m}$, $F(x)=x^{2^i+1}$, $\gcd(m,i)=1$, $G:\F_{2^m}\to \F_{2^m}$, $G(x)=x^{2^i+1}+(x^{2^i}+x+1)\Tr_m(x^{2^i+1})$. From the proof of \cite[Theorem 2]{BCP} we know that $F$ and $G$ are CCZ equivalent. For $m = 4$ and $i = 1$, using computers we have that ${_{cc}}\Delta_F\neq {_{cc}}\Delta_G$ for all $c \in \F_{2^4} \setminus \F_{2^2}$. Indeed, for $c \in \F_{2^4} \setminus \F_{2^2}$, ${_{cc}}\Delta_F = 3$ and ${_{cc}}\Delta_G = 4$.\\
(ii) For $m$ divisible by 6, the function $H:\F_{2^m}\to \F_{2^m}$ defined by $H(x)=\big(x+\Tr_{m}^3(x^{2(2^i+1)}+x^{4(2^i+1)})+\Tr_m(x)\Tr_{m}^3(x^{2^i+1}+x^{2^{2i}(2^i+1})\big)^{2^i+1}$ is CCZ-equivalence to $F(x)=x^{2^i+1}$ with $\gcd(m,i)=1$ \cite[the proof of Theorem 3]{BCP}. For $m = 6$ and $i = 1$ we have that ${_{cc}}\Delta_F\neq {_{cc}}\Delta_H$ for all $c \in \F_{2^6} \setminus \F_{2}$. Indeed, for $c \in \F_{2^6} \setminus \F_{2}$, ${_{cc}}\Delta_F = 3$ and ${_{cc}}\Delta_H \in \{5,6,7,8,9\}$.
\end{example}

\medskip 
The authors of \cite{BKM} shows that $F$ and
$F\circ A$ have the same $c$-differential uniformity where $A$ is
an affine permutation. But $F$ and $A\circ F$ do not have the same
$c$-differential uniformity in generally, and we already found a
counterexample in \cite{JKK22}. The following argument shows that $c$-differential uniformity is preserved under $c1$-equivalent and hence recovers \cite[Theorem 3.2]{BKM}.



\begin{lem}\label{cdu_affine_prop}
Let $c\in \F_{p^{\gcd(n,s)}}^\times$.  If  $F,F': \F_{p^n}\to \F_{p^s}$ are $c1$-equivalent
functions, then one has $_c\Delta_F={}_c\Delta_{F'}$. In
particular,
\begin{enumerate}
\item if $F'=F\circ A$ for some affine permutation $A:\F_{p^n}\to \F_{p^n}$ then $_c\Delta_F={}_c\Delta_{F'}$,
\item if $F$ and $F'$ are $c$-affine equivalent, then one
has $_c\Delta_F={}_c\Delta_{F'}$.
\end{enumerate}
\end{lem}
\begin{proof}
Let $F=A_1 \circ F' \circ A_2$ where $A_1:\F_{p^s}\to \F_{p^s}$ is a $c$-affine permutation and $A_2:\F_{p^n}\to \F_{p^n}$ is an affine permutation. We present $A_i(x) = L_i (x)+v_i$ with
linearized permutation $L_i(x)$  and $v_i=A_i(0)$ for $i=1,2$.
Recall that $c1$-equivalence of $F$ and $F'$ implies that
$$A_1(cx)=cA_1(x)-(c-1)v_1, \qquad A_i(x+a)=A_i(x)+A_i(a)-v_i \quad
\textrm{for}\,\,\,  i=1,2.$$
 Letting $F''=F'\circ A_2$, one has
$$b=F''(x+a)-cF''(x)=F'(A_2(x+a))-cF'(A_2(x))=F'(A_2(x)+A_2(a)-v_2)-cF'(A_2(x)).$$
Therefore,  letting $y=A_2(x)$ and $a_2=A_2(a)-v_2$, the solution
$x$ of $b=F''(x+a)-cF''(x)$ and  the solution $y$ of
$b=F'(y+a_2)-cF'(y)$ has one to one correspondence because $A_2$
is a permutation, which implies that
 $_c\Delta_{F'}={}_c\Delta_{F''}$.
Now we will show $_c\Delta_{F}={}_c\Delta_{F''}$ where $F=A_1\circ
F''$. From $_cD_aF(x)=b$, we have
\begin{align*}
b&=F(x+a)-cF(x)=A_1(F''(x+a))-cA_1(F''(x)) \\
&= A_1 (F''(x+a))-A_1(cF''(x))
+(1-c)v_1=A_1(F''(x+a)-cF''(x))-cv_1,
\end{align*}
and the above equation is equivalent to
 $A_1^{-1}(b+cv_1) = F''(x+a)-cF''(x)={}_cD_aF''(x).$
Therefore the number of solution of $_cD_aF(x)=b$ is equal to the
number of solutions of $_cD_aF''(x)=A_1^{-1}(b+cv_1)$.
\end{proof}

\section{Characterizations of $cc$-differential uniformity in terms of the Walsh transforms}
Let $\xi_p:=e^{\frac{2\pi i}{p}}$ be the complex primitive $p^{th}$ root of unity. For a function $F:\F_{p^n}\to \F_p$, its \textit{Walsh-Hadamard transform} $\cW_F$ is the Fourier transform of the function $\xi_p^{F(x)}$, i.e.
$$\cW_F(u):=\sum_{x\in \F_{p^n}}\xi_p^{F(x)-\Tr_n(ux)} \mbox{ for every } u\in \F_{p^n}.$$

 For a function $F:\F_{p^n}\to \F_{p^s}$, we define its \textit{Walsh transform} $\cW_F:\F_{p^n}\times \F_{p^s}\to \C$ by
$$\cW_F(u,v):=\sum_{x\in \F_{p^n}}\xi_p^{\Tr_s(vF(x))-\Tr_n(ux)} \mbox{ for every } u\in \F_{p^n},v\in \F_{p^s},$$
i.e. $W_F(u,v)$ is the value at $u$ of the Walsh-Hadamard transformation of the function $\Tr_s(vF(x))$. We denote by $\overline{\cW_F}(u,v)$ the complex conjugate of $\cW_F(u,v)$.

Given two functions $f,g:2^{p^t}\to \R$, we denote $f\otimes g$ the convolution product
$$(f\otimes g)(a):=\sum_{x\in 2^{p^t}}f(x)g(x+a) \mbox{ for every } a\in 2^{p^t}.$$
\begin{lem}
Let $m,n,s\in \N$, $c\in \F_{p^{\gcd(n,s)}}^\times$, and $F$ be a function from $\F_{2^n}$ to $\F_{2^s}$. We define the function $F_c:\F_{2^n}\to \F_{2^s}$ by $F_c(x)=F(cx)$. Then for every $(u,v)\in \F_{2^n}\times \F_{2^s}$ we have
\begin{align*}
{_{cc}}\Delta_F(u,v)=1_{\cG_{cF}}\otimes 1_{\cG_{F_c}}(uc^{-1},v),
\end{align*}
where $1_A$ is the characteristic function of a given set $A\subset \F_{2^n}\times \F_{2^n}$, i.e. $1_A(x,y)$ is $1$ if $(x,y)\in A$, and is $0$ otherwise.
\end{lem}
\begin{proof}
We have
\begin{align*}
1_{\cG_{cF}}\otimes 1_{\cG_{F_c}}(uc^{-1},v)&=\sum_{x\in 2^{p^n},y\in 2^{p^s}}1_{\cG_{cF}}(x,y)1_{\cG_{F_c}}(x+uc^{-1},y+v)\\
&=\sum_{x\in 2^{p^n}}1_{\cG_{F_c}}(x+uc^{-1},cF(x)+v)\\
&={_{cc}}\Delta_F(u,v).
\end{align*}
\end{proof}

In the next theorem, we extend \cite[Theorem 1]{Carlet18} to the $cc$-differential context. Note that its version for $c$-differential uniformity was also proved in \cite{EFRST}.
\begin{thm}
\label{T-characterizations in terms of the Walsh transform}
Let $m,n,s\in \N$ and $F$ be a function from $\F_{p^n}$ to $\F_{p^s}$. Let $c\in \F_{p^{\gcd(n,s)}}^\times$ and let $\varphi_{m}=\sum_{k\geq 0}A_kx^k$ be a polynomial over $\R$ such that $\varphi_{m}(x)=0$ for all $x\in \N$ with $x\leq m$ and $\varphi_m(x)>0$ for $x\in \N$ with $x>m$. Then we have
\begin{align}
\label{F-characterizations in terms of the Walsh transform}
p^{2n}A_0+\sum_{k\geq 1}p^{-(n+s)k}A_kG_{k+1}\geq 0,
\end{align}

where \begin{align*}
G_{k+1}:=&\sum_{\substack{v_1,\dots,v_k\in \F_{p^n}\\u_1,\dots, u_k\in \F_{p^s}}}\overline{\cW_F}(\sum_{j=1}^ku_i,\sum_{j=1}^kv_i)\cW_F(c\sum_{j=1}^ku_j,c\sum_{j=1}^kv_j)\\
&\cdot\prod_{j=1}^k\overline{\cW_F}(cu_j,cv_j)\cW_F(u_j,v_j).
\end{align*}
And the equality holds if and only if ${_{cc}}\Delta_F=m$.
\end{thm}
\begin{proof}
Given $a,b\in \F_{p^n}$, we define $S_F(a,b,c):=|\{x\in \F_{p^n}: {_{cc}}D_aF(x)={_{cc}}D_aF(b)\}|$. From our assumption on the polynomial $\varphi_m$, for every $a,b\in \F_{p^n}$ we get that
$$\sum_{k\geq 0}A_kS^k_F(a,b,c)\geq 0,$$
and the equality holds if and only if ${_{cc}}\Delta_F(a,{_{cc}}D_aF(b))\leq m$. Hence
\begin{align}
\label{F-temp1}
\sum_{k\geq 0}\sum_{a,b\in \F_{p^n}}A_kS^k_F(a,b,c)\geq 0,
\end{align}
with equality if and only if ${_{cc}}\Delta_F\leq m$.

As $\sum_{v\in \F_{p^s}}\xi_p^{\Tr_s(vx)}$ is $p^s$ if $x=0$ and otherwise equals to $0$, we get that
$$ S_F(a,b,c)=p^{-s}\sum_{x\in \F_{p^n},v\in \F_{p^s}}\xi_p^{\Tr_s(v({_{cc}}D_aF(x)-{_{cc}}D_aF(b)))}.$$
Therefore for every $k\geq 1$ we have that
\begin{align*}
\sum_{a,b\in \F_{p^n}}S_F^k(a,b,c)&=p^{-ks}\sum_{a,b\in \F_{p^n}}\sum_{\substack{x_1,\dots,x_k\in \F_{p^n}\\ v_1,\dots,v_k\in \F_{p^s}}}\xi_p^{\sum_{j=1}^k\Tr_s(v_j({_{cc}}D_aF(x_j)-{_{cc}}D_aF(b)))}\\
&=p^{-ks}\sum_{a,b\in \F_{p^n}}\sum_{\substack{x_1,\dots,x_k\in \F_{p^n}\\ v_1,\dots,v_k\in \F_{p^s}}}\xi_p^{\sum_{j=1}^k\Tr_s(v_j(F(cx_j+a)-cF(x_j)-F(cb+a)+cF(b))}.
\end{align*}
As $\sum_{u_0\in \F_{p^n}}\xi_p^{\Tr_n(u_0(d-cb-a))}=p^n$ if $d=cb+a$ and 0, otherwise, and $\sum_{u_j\in \F_{p^n}}\xi_p^{\Tr_n(u_j(cx_j+a-y_j))}=p^n$ if $y_j=cx_j+a_j$ and 0, otherwise, we deduce that
\begin{align*}
\sum_{a,b\in \F_{p^n}}S_F^k(a,b,c)=&p^{-ks}p^{-(k+1)n}\sum_{a,b,d\in \F_{p^n}}\\
&\sum_{\substack{x_1,\dots,x_k\in \F_{p^n}\\y_1,\dots,y_k\in \F_{p^n}\\ v_1,\dots,v_k\in \F_{p^s}\\u_0,u_1,\dots, u_k\in \F_{p^n}}}\xi_p^{\sum_{j=1}^k\Tr_s(v_j(F(y_j)-cF(x_j)-F(d)+cF(b))+\Tr_n(u_j(cx_j+a-y_j))+\Tr_n(u_0(d-cb-a))}\\
=&p^{-ks}p^{-(k+1)n}\sum_{\substack{v_1,\dots,v_k\in \F_{p^s}\\u_0,u_1,\dots, u_k\in \F_{p^n}}}\cW_F(cu_0,c\sum_{j=1}^kv_j)\overline{\cW_F}(u_0,\sum_{j=1}^kv_j)\prod_{j=1}^k\overline{\cW_F}(cu_j,cv_j)\cW_F(u_j,v_j)\\
&\cdot \sum_{a\in \F_{p^n}}\xi_p^{\Tr_n(a\sum_{j=1}^j u_j-u_0)}\\
=&p^{-ks}p^{-(k+1)n}p^n\sum_{\substack{v_1,\dots,v_k\in \F_{p^s}\\u_1,\dots, u_k\in \F_{p^n}}}\cW_F(c\sum_{j=1}^ku_j,c\sum_{j=1}^kv_j)\\
&\cdot\overline{\cW_F}(\sum_{j=1}^ku_i,\sum_{j=1}^kv_i)\prod_{j=1}^k\overline{\cW_F}(cu_j,cv_j)\cW_F(u_j,v_j)\\
=&p^{-(n+s)k}G_{k+1}.
\end{align*}
On the other hand, $\sum_{a,b\in \F_{p^n}}S_F^0(a,b,c)=p^{2n}$, hence combining with \eqref{F-temp1} we get the result.
\end{proof}
\begin{remark}
Given $m\in \N$, the polynomial $\varphi_m(x)=\prod_{j=1}^m(x-j)$ satisfies the assumptions on Theorem \ref{T-characterizations in terms of the Walsh transform}.
\end{remark}
If $m=1$ and $\varphi_1(x)=x-1$ then \eqref{F-characterizations in terms of the Walsh transform} becomes
$$-p^{2n}+p^{-(s+n)}\sum_{u\in \F_{p^n}, v\in \F_{p^s}}|\cW_F(u,v)|^2|\cW_F(cu,cv)|^2\geq 0.$$
\begin{cor}
Let $n,s\in \N$, $c\in \F_{p^{\gcd(n,s)}}^\times$, and $F:\F_{p^n}\to \F_{p^s}$ be a function. Then
$$\sum_{u\in \F_{p^n}, v\in \F_{p^s}}|\cW_F(u,v)|^2|\cW_F(cu,cv)|^2\geq p^{3n+s}.$$
The equality holds if and only if $F$ is PccN.
\end{cor}
If $m=2$ and $\varphi_2(x)=(x-1)(x-2)=x^2-3x+2$ then \eqref{F-characterizations in terms of the Walsh transform} will be
\begin{align*}
2&p^{2n}-3p^{-(s+n)}\sum_{u\in \F_{p^n}, v\in \F_{p^s}}|\cW_F(u,v)|^2|\cW_F(cu,cv)|^2\\
&+p^{-2(s+n)}\sum_{u_1,u_2\in \F_{p^n}, v_1,v_2\in \F_{p^s}}\cW_F(c(u_1+u_2),c(v_1+v_2))\overline{\cW_F}(u_1+u_2,v_1+v_2)\\
&\cdot \overline{\cW_F}(cu_1,cv_1)\overline{\cW_F}(cu_2,cv_2)\cW_F(u_1,v_1)\cW_F(u_2,v_2)\geq 0.
\end{align*}
Therefore we get the following corollary.
\begin{cor}
Let $n,s\in \N$, $c\in \F_{p^{\gcd(n,s)}}^\times$ and $F:\F_{p^n}\to \F_{p^s}$ be a function. Then
\begin{align*}
&\sum_{u_1,u_2\in \F_{p^n}, v_1,v_2\in \F_{p^s}}\cW_F(c(u_1+u_2),c(v_1+v_2))\overline{\cW_F}(u_1+u_2,v_1+v_2) \\
&\cdot\overline{\cW_F}(cu_1,cv_1)\overline{\cW_F}(cu_2,cv_2)\cW_F(u_1,v_1)\cW_F(u_2,v_2)\\
&\geq 3p^{s+n}\sum_{u\in \F_{p^n}, v\in \F_{p^s}}|\cW_F(u,v)|^2|\cW_F(cu,cv)|^2-2p^{2(s+2n)}.
\end{align*}
And the equality holds if and only if $F$ is APccN.
\end{cor}
\begin{lem}
Let $n,s\in \N$ and $F$ be a function from $\F_{p^n}$ to $\F_{p^s}$ and let $c\in \F_{p^{\gcd(n,s)}}^\times$. Then for every $a\in \F_{p^n}$ and $k\geq 1$ we have
\begin{align*}
\sum_{b\in \F_{p^n}}S_F^k(a,b,c)=p^{-ks}\sum_{v_1,\dots,v_k\in \F_{p^s}}\overline{\cW}_{_{cc}D_aF}(0,\sum_{j=1}^kv_j)\prod_{j=1}^k{\cW}_{_{cc}D_aF}(0,v_j).
\end{align*}
\end{lem}
\begin{proof}
As $ S_F(a,b,c)=p^{-s}\sum_{x\in \F_{p^n},v\in \F_{p^s}}\xi_p^{\Tr_s(v({_{cc}}D_aF(x)-{_{cc}}D_aF(b)))}$, for every $k\geq 1$ we get that
\begin{align*}
\sum_{b\in \F_{p^n}}S_F^k(a,b,c)&=p^{-ks}\sum_{b\in \F_{p^n}}\sum_{\substack{x_1,\dots,x_k\in \F_{p^n}\\ v_1,\dots,v_k\in \F_{p^s}}}\xi_p^{\sum_{j=1}^k\Tr_s(v_j({_{cc}}D_aF(x_j)-{_{cc}}D_aF(b)))}\\
&=p^{-ks}\sum_{\substack{b,x_1,\dots,x_k\in \F_{p^n}\\ v_1,\dots,v_k\in \F_{p^s}}}\xi_p^{\sum_{j=1}^k\Tr_s(v_j\cdot{_{cc}}D_aF(x_j))-\Tr_s(\sum_{j=1}^kv_j\cdot{_{cc}}D_aF(b))}\\
&=p^{-ks}\sum_{v_1,\dots,v_k\in \F_{p^s}}\overline{\cW}_{{_{cc}}D_aF}(0, \sum_{j=1}^kv_j)\prod_{j=1}^k\cW_{{_{cc}}D_aF}(0,v_j).
\end{align*}
\end{proof}

The following theorem is an extension of \cite[Theorem 3.8]{Carlet19} to all prime number $p$ as well as to our cc-differential context.
\begin{thm}
\label{T-characterizations in terms of the Walsh transform of cc-derivatives}
Let $m,n,s\in \N$, $c\in \F_{p^{\gcd(n,s)}}^\times$, and $F$ be a function from $\F_{p^n}$ to $\F_{p^s}$. Let $\varphi_{m}=\sum_{i\geq 0}A_kx^k$ be a polynomial over $\R$ such that $\varphi_{m}(x)=0$ for all $x\in \N$ with $x\leq m$ and $\varphi_m(x)>0$ for $x\in \N$ with $x>m$. Then for every $a\in \F_{p^n}$ ($a\neq 0$ if $c=1$), we have
\begin{align*}
p^nA_0+\sum_{k\geq 1}p^{-ks}\sum_{v_1,\dots,v_k\in \F_{p^s}}\overline{\cW}_{{_{cc}}D_aF}(0, \sum_{j=1}^kv_j)\prod_{j=1}^k\cW_{{_{cc}}D_aF}(0,v_j)\geq 0.
\end{align*}
The equality holds if and only if $F$ is $cc$-differential $m$-uniform.
\end{thm}
\begin{proof}
For $k=0$, we have $\sum_{b\in \F_{p^n}}S_F^k(a,b,c)=p^n$. Therefore, using the same method as the proof of Theorem \eqref{T-characterizations in terms of the Walsh transform} we get the result.
\end{proof}
\begin{remark}
When $n=s$, $p=2$, $c=1$, and $\varphi_2=x-2$, our Theorem \ref{T-characterizations in terms of the Walsh transform of cc-derivatives} recovers \cite[Theorem 2]{BCCL}.
\end{remark}
\section{Properties of $cc$-differential uniformity}
In this section we will present several basic properties of $cc$-differential uniformity, and illustrate examples to show differences between $c$-differential uniformity and $cc$-differential uniformity. In subsection 4.1, we investigate $cc$-differential uniformity of power functions  $F(x)=x^d$, and in subsection 4.2, we study for the case $c=-1$.
\begin{lem}
Let $n,s\in \N$, $c\in \F_{p^{\gcd(n,s)}}^\times $ and $F: \F_{p^n}\to \F_{p^s}$. Then for every $a\in \F_{p^n}, b\in \F_{p^s}$ we have
\begin{enumerate}
\item ${_{cc}}\Delta_F(a,b)=\big|\bigcup_{y\in \F_{p^s}}A_{c^{-1}y}\cap (c^{-1}A_{y+b}-c^{-1}a)\big|,$
where $A_z:=F^{-1}(z)$ for every $z\in \F_{p^s}$;
\item  ${_{cc}}\Delta_{F}(a,b)={_{\frac1c\frac1c}}\Delta_{F}(-\frac{a}{c},-\frac{b}{c}
)$.
\end{enumerate}
\end{lem}
\begin{proof}
1) We have
\begin{align*}
{_{cc}}\Delta_F(a,b)&=|\{x\in \F_{p^n}: F(cx+a)-cF(x)=b\}|\\
&=\big|\bigcup_{y\in \F_{p^s}}\{x\in \F_{p^n}:F(x)=c^{-1}y \mbox{ and } F(cx+a)=y+b\}\big|\\
&=\big|\bigcup_{y\in \F_{p^s}}A_{c^{-1}y}\cap (c^{-1}A_{y+b}-c^{-1}a)\big|.
\end{align*}
2) Letting $y=cx+a$, one gets $
b=F(cx+a)-cF(x)=F(y)-cF(\frac{y-a}{c}). $ Therefore, there is one
to one correspondence between the solutions of
$F(\frac{1}{c}y-\frac{a}{c})-\frac{1}{c}F(y)=-\frac{b}{c}$ and the
solutions of $F(cx+a)-cF(x)=b$.

\end{proof}

\begin{definition}

     A function $F : \Fpn \to \Fpn$ is called a generalized
 DO (Dembowski-Ostrom) polynomial of weight $k\geq 1$ and type $(n_1,n_2,\cdots, n_k) \in \left(\mathbb Z^\times\right)^k $ if
 $$F(x) = \displaystyle \sum_{ 0 \leq i_1, i_2,\cdots , i_k  < n } a_{i_1i_2\cdots i_k}x^{n_1p^{i_1}+n_2p^{i_2}+\cdots +n_kp^{i_k}} \in \Fpn[x]$$
where the exponent $n_1p^{i_1}+n_2p^{i_2}+\cdots +n_kp^{i_k}$ is
evaluated up to $\pmod{p^n-1}$.
\end{definition}
\begin{remark} \hfill

 \begin{enumerate}
 \item[--] A DO polynomial of weight $1$ and type $(1)$ is a $\F_p$-linearized polynomial   $\displaystyle \sum_{ 0 \leq i < n }
 a_{i}x^{p^i}$.
  \item[--] A DO polynomial of weight $2$ and type $(1,1)$ is the original Dembowski-Ostrom polynomial
   $\displaystyle \sum_{ 0 \leq i,j < n }
 a_{ij}x^{p^i+p^j}$.
  \end{enumerate}
\end{remark}

 \begin{lem}\label{DO}
  Let  $c\in
\F_{p}^\times$ and let $F(x) = \displaystyle \sum_{ 0 \leq i_1,
i_2,\cdots , i_k  < n } a_{i_1i_2\cdots
i_k}x^{n_1p^{i_1}+n_2p^{i_2}+\cdots +n_kp^{i_k}} \in
  \Fpn[x]$ be a DO polynomial of weight $k$ and type $(n_1,n_2,\cdots,
  n_k).$
  Then, letting $c'=c^{1-\sum_{s=1}^k n_s}=c^{1-n_1-n_2 \cdots
 -n_k}$, one has ${_{cc}}\Delta_F = {_{c'}}\Delta_{F}$.
 \end{lem}
\begin{proof}
   Replacing $x$ with $\frac{x}{c}$, the difference equation  $F(cx+a) - cF(x) =
   b$ is equivalent to $F(x+a) - cF(\frac{x}{c}) = b$.
     For $c\in \F_{p}^\times$, one has
\begin{align*}
F(cx) &= \displaystyle \sum_{ 0 \leq i_1, i_2,\cdots , i_k  < n }
a_{i_1i_2\cdots i_k}(cx)^{n_1p^{i_1}+n_2p^{i_2}+\cdots
+n_kp^{i_k}}\\
 &=\displaystyle \sum_{ 0 \leq i_1, i_2,\cdots , i_k
< n } a_{i_1i_2\cdots i_k}c^{\sum_{s=1}^k n_s}
x^{n_1p^{i_1}+n_2p^{i_2}+\cdots +n_kp^{i_k}} \quad (\because
c^p=c) \\
&=c^{\sum_{s=1}^k n_s}F(x).
\end{align*}
Therefore,
$$
b=F(x+a) - cF\left(\frac{x}{c}\right)=F(x+a)-c\cdot
c^{-\sum_{s=1}^k n_s}F(x)=F(x+a)-c^{1-\sum_{s=1}^k n_s}F(x),
$$
which implies that
 ${_{cc}}\Delta_F = {_{c'}}\Delta_{F}$ with $c'=c^{1-\sum_{s=1}^k n_s}=c^{1-n_1-n_2 \cdots
 -n_k}$.
\end{proof}

Given $s,t\in \N$ such that $t|s$. We say that a map $L: \F_{q^s}\to \F_{q^s}$ is $\F_{q^t}$-linearized if $L(x+y)=L(x)+L(y)$ for every $x,y\in \F_{q^s}$ and $L(ax)=aL(x)$ for every $a\in \F_{q^t}$, $x\in \F_{q^s}$.
The following lemma is a version of \cite[Theorem 6]{LRS} for cc-differentials.
\begin{lem}\label{pant}
Let $q=p^n$, $t,s\in \N$ such that $t|s$, $c\in \F_{q^t}^\times$ and $F:\F_{q^s}\to
\F_{q^s}$. Let $u,v\in\F_{q^s}$ with $\Tr^{q^t}_{q^s}(-uv)\neq 1$ and we define the map $G:\F_{q^s}\to \F_{q^s}$ by
$$G(x)=F(x)+u\Tr^{q^t}_{q^s}(vF(x)).$$
Then ${_{cc}}\Delta_{G}={_{cc}}\Delta_{F}$. In particular, if $F$ is a PccN function then so is $G$.
\end{lem}
\begin{proof}
Let $h(x)=x+u\Tr^{q^t}_{q^s}(vx)$. Then, it is straightforward to
show that $h$ is a $\F_{q^t}$-linearized. Furthermore, $h$ is a permutation. That is,
$h(x)=h(y)$ with $x\neq y$ implies
\begin{align*}
x-y=ku \text{ with } 0\neq k=\Tr^{q^t}_{q^s}(v(y-x)) \Leftrightarrow
0\neq k=\Tr^{q^t}_{q^s}(-kuv) \Leftrightarrow  1=\Tr^{q^t}_{q^s}(-uv).
\end{align*}
Therefore $G(x)=h\circ F(x)$ and $F(x)$ are $c$-affine
equivalent for $c\in \F_{q^t}$ and
${_{cc}}\Delta_{G}={_{cc}}\Delta_{F}$.
\end{proof}

Next we will illustrate examples to show differences between $c$-differential uniformity and $cc$-differential uniformity.
\begin{example}
        Let $F(x) = x + \Tr(x^3)$ over $\F_{2^6}$. Then using computers, we get
        \begin{equation*}
            \begin{split}
                \{ _{cc}\Delta_F: c \in \F_{2^6} \setminus \F_2 \} &= \{40^2, 26^{36}, 28^{24}\},\\
                \{ {_c}\Delta_F : c \in \F_{2^6} \setminus \F_2 \} &=\{1^2, 2^{60}\}.
            \end{split}
        \end{equation*}
    \end{example}
\begin{example}
 Let $q=p^n$, $t,s\in \N$ such that $t|s$, $c\in \F_{q^t}$ and $F:\F_{q^s}\to
\F_{q^s}$. Let $u,v\in
\F_{q^s}$ such that $\Tr_{q^s}^{q^t}(uv)\neq -1$. Let $L_1,L_2:
\F_{q^s}\to \F_{q^s}$ be $\F_{q^t}$-linearized functions and
$\rho\in \F_{q^s}^times$ such that $L_1$ is a permutation and
$L_2(\Tr_{q^s}^{q^t}(\rho))=0$. We define function $G:\F_{q^s}\to \F_{q^s}$ by
$$G(x)=L_1(x)+L_2(\rho)\Tr_{q^s}^{q^t}(L_2(x)).$$
Since $G(x)$ is linear, by direct
computation, one can show that ${_{cc}}\Delta_{G}=q^s$. However, ${_{c}}\Delta_{G}=1$ if $c\neq 1$ \cite[Theorem 6]{LRS}.
\end{example}

\subsection{The power functions $F(x)=x^d$}
In this subsection we investigate $cc$-differential uniformity of power functions $F(x)=x^d$ on $\F_{p^n}$.
\begin{lem}
\label{L-relations between $c$-DU and $cc$-DU for monomials}
Let $n,s\in \N$. Let $F(x) = x^d$ be a function on $\Fpn$ and $a,b \in \Fpn$, $c \in \Fpn^\times$. We have that
\begin{enumerate}
\item ${_{cc}}\Delta_F(a,b) = {_{c^{1-d}}}\Delta_F(\frac{a}{c},\frac{b}{c^d});$
\item if $c^{d-1}\neq 1$ then ${_{cc}}\Delta_F=\max\{{_{cc}}\Delta_F(1,b):b\in \F_{p^n}\}\cup \{\gcd(d,p^n-1)\}$ ;
\item if $c^{d-1}=1$ and $c\neq 1$ then ${_{cc}}\Delta_F=p^n$;
\item if $c=1$ then ${_{cc}}\Delta_F=\max\{{_{cc}}\Delta_F(1,b):b\in \F_{p^n}\}$.
\end{enumerate}
\end{lem}
\begin{proof}
1) It follows from that $x$ is a solution of the equation $(cx+a)^d - cx^d = b$ if and only if $x$ is a solution of the equation $(x+\frac{a}{c})^d - c^{1-d}x^d = \frac{b}{c^d}$.

2) If $a\neq 0$, the equation $(cx+a)^d-cx^d=b$ is equivalent to the equation $(c\dfrac{x}{a}+1)^d-c(\dfrac{x}{a})^d=\dfrac{b}{a^d}$. If $a=0$, we have that ${_{cc}}\Delta_F(0,b)$ is the number of solutions of the equation $(cx)^d-cx^d=b$. As $c^{d-1}\neq 1$ we have that 
\begin{equation*}
{_{cc}}\Delta_F(0,b)= \left\{
        \begin{array}{ll}
            1 & \quad \mbox{ if }b=0, \\
           \gcd(d,p^n-1) & \quad \mbox{ if }  \dfrac{b}{c^d-c}\in \F_{p^n}^\times \mbox{ is a dth power},\\
           0 & \quad \mbox{ otherwise}.
        \end{array}
    \right.
\end{equation*}

3) As $c^{d-1}=1$ we have that ${_{cc}}\Delta_F(0,0)=p^n$ and hence ${_{cc}}\Delta_F=p^n$ for $c\neq 1$.

4) As $c=1$ we have that ${_{cc}}\Delta_F=\max\{{_{cc}}\Delta_F(a,b):a,b\in \F_{p^n}, a\neq 0\}=\max\{{_{cc}}\Delta_F(1,b):b\in \F_{p^n}\}$. 
\end{proof}

In \cite{EFRST, MRSYZ}, the authors studied the $c$-differential uniformity of the Gold function on $\F_{p^n}$, $x\mapsto x^{p^k+1}$. In the following theorem we study $cc$-differential uniformity for a such function. In particular, we show that APccN functions can be obtained in these functions for every $p>2$ and $c\in \F_{p}\setminus\{1\}$.
\begin{thm}
Let  $F(x)=x^d$ be a power function on $\mathbb F_{p^n}$, where $d=p^m+1$ for some $m\in \N$. For $1\neq c\in \F_{p^{\gcd(m,n)}}^\times$, we have that ${_{cc}}\Delta_F=\gcd(d,p^n-1)$. In particular,
\begin{enumerate}
\item for $p=2$, the $cc$-differential uniformity of $F$ is $\dfrac{2^{\gcd(2m,n)-1}}{2^{\gcd(m,n)-1}}$;
\item if $p>2$ and $\dfrac{n}{\gcd(n,m)}$ is odd then ${_{cc}}\Delta_F=2$;
\item if $p>2$ and $\dfrac{n}{\gcd(n,m)}$ is even then ${_{cc}}\Delta_F=p^{\gcd(m,n)}+1$.
\end{enumerate}
\end{thm}
\begin{proof}
We have that
\begin{align*}
(cx+1)^d-cx^d&=(c^d-c)x^{p^m+1}+c^{p^m}x^{p^m}+cx+1\\
&=(c^2-c)(x^{p^m+1}+\dfrac{c}{c^2-c}x^{p^m}+\dfrac{c}{c^2-c}x)+1\\
&=(c^2-c)(x^{p^m+1}+\dfrac{1}{c-1}x^{p^m}+\dfrac{1}{c-1}x+\dfrac{1}{(c-1)^{p^m+1}})-\dfrac{c}{c-1}+1\\
&=(c^2-c)(x+\dfrac{1}{c-1})^{d}+\dfrac{1}{1-c}.
\end{align*}
Hence ${_{cc}}\Delta_F(1,b)=\#\big\{x\in \F_{p^n}:(c^2-c)(x+\dfrac{1}{c-1})^{d}=b-\dfrac{1}{1-c}\big\}$ Therefore $$\max\{{_{cc}}\Delta_F(1,b):b\in \F_{p^n}\}=\gcd(d,p^n-1).$$ Applying Lemma \ref{L-relations between $c$-DU and $cc$-DU for monomials} and \cite[Lemma 9]{EFRST} we get the result.
\end{proof}

Combining Lemma \ref{L-relations between $c$-DU and $cc$-DU for monomials} and results in \cite{EFRST,MRSYZ} we have the following table of some classes of functions $x^d$, for $c\in \F_{p^n}^\times $ with $c^{1-d}\neq 1$.
\begin{table}[H]

\begin{center}
   \footnotesize
\begin{tabular}{|c|c|c|c|c|c|}
\hline $d$ & $\F_{p^n}$ & ${_{cc}}\Delta_{F}={_{c^{1-d}}}\Delta_{F}$ & \mbox{Conditions } & \mbox{Ref} \\
\hline\hline 2 & every $p$ & 2 (APccN) & none & \cite[Theorem 10 i)]{EFRST}\\
\hline $2^n-2 $& $p=2$ & 2 (APccN) & $\Tr_n(c^{1-d})=\Tr_n(1/c^{1-d})=1$ & \cite[Theorem 12 ii)]{EFRST}  \\
\hline $2^n-2$ & $p=2$ & 3 & $\Tr_n(1/c^{1-d})=0$ or $\Tr_n(c^{1-d})=0$ & \cite[Theorem 12 iii)]{EFRST}\\
\hline $\dfrac{3^k+1}{2}$ & $p=3$ & 1 (PccN) & if and only if $\dfrac{2n}{\gcd(k,2n)}$ is odd (*)& \cite[Theorem 10 iii)]{EFRST}   \\
\hline $p^n-2$ & $p>2$ & 3 & \thead{$c^{1-d}\neq 4,4^{-1}; c^{2-2d}-4c^{1-d}\in [\F_{p^n}]^2$\\
 or $1-4c^{1-d}\in [\F_{p^n}]^2$} & \cite[Theorem 13 ii)]{EFRST}  \\
 \hline $p^n-2$ & $p>2$ & 2 (APccN) & $c^{1-d}= 4$ or $c^{1-d}=4^{-1}$ & \cite[Theorem 13 iii)]{EFRST} \\
  \hline $p^n-2$ & $p>2$ & 2 (APccN) &  $c^{2-2d}-4c^{1-d}\notin [\F_{p^n}]^2$ and $1-4c^{1-d}\in [\F_{p^n}]^2$ & \cite[Theorem 13 iv)]{EFRST}  \\
  \hline $2^k+1$ & $p=2$ & $2^{\gcd(n,k)}+1$  &  $2\leq k<n$, $n\geq 3$, $c^{1-d}\in \F_{2^n}\setminus \F_{2^d}$ & \cite[Theorem 4]{MRSYZ}\\
  \hline $\dfrac{p^k+1}{2}$ & $p>2$ & 1 (PccN)  &  $1\leq k<n$, $n\geq 3$, $c^{1-d}=-1$, $\dfrac{2n}{\gcd(2n,k)}$ is odd & \cite[Theorem 6]{MRSYZ}\\
  \hline $\dfrac{p^k+1}{2}$ & $p>2$ & $\dfrac{p^{\gcd(k,n)+1}}{2}$  &  $1\leq k<n$, $n\geq 3$, $c^{1-d}=-1$, $\dfrac{2n}{\gcd(2n,k)}$ is even & \cite[Theorem
  6]{MRSYZ} \\
   \hline
\end{tabular}
\caption{$cc$-differential uniformity of some classes of functions $x^d$.}
\end{center}
\vspace{-0.5cm}
\end{table}
      
(*): Note that there is a typo in the statement of \cite[Theorem 10 iii)]{EFRST} as $\dfrac{2n}{\gcd(k,n)}$ there would be $\dfrac{2n}{\gcd(k,2n)}$.

Similarly, using Lemma \ref{L-relations between $c$-DU and $cc$-DU for monomials} and results of $c$-differential uniformity of power functions on \cite{WZ,WLZ, Y,ZH}, we also can get corresponding results for $cc$-differential uniformity of such functions.

       \subsection{The case : c = -1}

When $p\neq 2$ and $c=-1$, then ${_{cc}}\Delta_{F}(a,b)$ is even for every $a,b$ with $b\neq
2F(\frac{a}{2})$. We have that because
$F(-y+a)+F(y)=F(\frac{1}{c}y-\frac{a}{c})-\frac{1}{c}F(y)=-\frac{b}{c}=b$
with $y=-x+a$. In other words, if $x$ is a solution of
$F(-x+a)+F(x)=b$, then $a-x$ is also a solution, and $x=a-x$
happens exactly when $x=\frac{a}{2}$ so that $b=2F(\frac{a}{2})$. Hence we get the following lemma.

\begin{lem}\label{-1PccN cor}
Let $p\neq 2$ and $c=-1 \in \mathbb F_{p^n}$. Then a function $F:\mathbb F_{p^n}:\F_{p^n}$ is PccN if and only if, for every $a, b\in
\mathbb F_{p^n}$, the difference equation $F(a-x)+F(x)=b$ has no
solution when $b\neq 2F(\frac{a}{2})$ and $\frac{a}{2}$ is the
only solution when $b= 2F(\frac{a}{2})$.
\end{lem}

\begin{example}
Let $p\neq 2$ and $c=-1 \in \mathbb F_{p^n}$. Then $F(x)=x^2$ is
APccN : since
$(a-x)^2+x^2=F(a-x)+F(x)=2F(\frac{a}{2})=\frac{a^2}{2}$, one has
$2x^2-2ax+a^2=\frac{a^2}{2}$ or equivalently
$\left(x-\frac{a}{2}\right)^2=0$. However replacing
$\frac{a^2}{2}$ with $\frac{a^2}{2}+b$, one gets
$\left(x-\frac{a}{2}\right)^2=b$ so that there exist two solutions
if $b$ is a quadratic residue in $\mathbb F_{p^n}$.
\end{example}

        \begin{cor}
            Let $p \neq 2$ and $c = -1$. There is no PccN function on $\Fpn$.
        \end{cor}
    \begin{proof}
        Let $F$ be a PccN function on $\Fpn$. Since $\bigcup_{b \in \Fpn} \{x \in \Fpn : F(cx+a) - cF(x) = b\} = \Fpn,$ it follows that $\displaystyle \sum_{b \in \Fpn} {_{cc}}\Delta_F(a,b) = p^n$ for each $a \in \Fpn$, $\mbox{ so that } {_{cc}}\Delta_F(a,b) = 1 \mbox{ for all } a,b \in \Fpn.$ This contradicts to Lemma \ref{-1PccN cor}.
    \end{proof}
\ \\

        \begin{lem}
        Let $p \neq 2$ and $c = -1$. Let $O$ be an odd function and $E$ be an even function on $\Fpn$. That is, it satisfies that $O(-x) = -O(x)$ and $E(-x) = E(x)$ for all $x \in \Fpn$. Then the followings are satisfied:
        \begin{enumerate}
            \item $_{cc}\Delta_{O}(a,b) = \Delta_{O}(-a,-b)$. Therefore $_{cc}\Delta_{O} = \Delta_{O}$.
            \item $_{cc}\Delta_{E}(a,b) = {_{c}}\Delta_{E}(-a,b)$. Therefore $_{cc}\Delta_{E} = {_{c}}\Delta_{E}$.
            \item Let $F(x) = E(x) + A(x)$ where $A:\Fpn\to \F_pn$ is an affine function. Then $_{cc}\Delta_{F} = {_{c}}\Delta_{E}$.
        \end{enumerate}
        \end{lem}
    \begin{proof}
    For (1), $x$ is a solution of the equation $O(-x+a) + O(x) = b$ if and only if $x$ is a solution of the equation $O(x-a) - O(x) = -b$ because $O(-x+a) = -O(x-a)$.     For (ii), $x$ is a solution of the equation $E(-x+a) + E(x) = b$ if and only if $x$ is a solution of the equation $E(x-a) - (-1) E(x) = b$ because $E(-x+a) = E(x-a)$. For (iii), as $F$ and $E$ are $c$-EA equivalent, applying the second claim in this lemma, we get the result.
    \end{proof}


\begin{thebibliography}{00}

\bibitem{BKM} D. Bartoli, L. Kölsch and G. Micheli, Differential biases, c-differential uniformity, and their relation to differential attacks, arXiv:2208.03884.
\bibitem{BT}D. Bartoli and M. Timpanella, On a generalization of planar functions. J. Algebraic Combin. 52 (2020), no. 2, 187--213.
\bibitem{BCCL} T. Berger, A. Canteaut, P. Charpin, Pascale and Y. Laigle-Chapuy, On almost perfect nonlinear functions over $\F_n^2$. IEEE Trans. Inform. Theory 52 (2006), no. 9, 4160--4170.
\bibitem{BC}L. Budaghyan, and C. Carlet, CCZ-equivalence of single and multi output Boolean functions. Finite fields: theory and applications, 43--54, Contemp. Math., 518, Amer. Math. Soc., Providence, RI, 2010.
\bibitem{BCP}L. Budaghyan, C. Carlet and A. Pott, New classes of almost bent and almost perfect nonlinear polynomials, in IEEE Transactions on Information Theory, vol. 52, no. 3, pp. 1141-1152, March 2006.
\bibitem{BH1} L. Budaghyan and T. Helleseth, New commutative semifields defined by new PN multinomials. Cryptogr. Commun. 3 (2011), no. 1, 1--16.
\bibitem{BH} L. Budaghyan and T. Helleseth, On isotopisms of commutative presemifields and CCZ-equivalence of functions. Internat. J. Found. Comput. Sci. 22 (2011), no. 6, 1243--1258.
\bibitem{Carlet} C. Carlet, Boolean Functions for Cryptography and Coding Theory 1st edition, Cambridge University Press 2020.
\bibitem{Carlet18} C. Carlet, Characterizations of the differential uniformity of vectorial functions by the Walsh transform, IEEE Trans. Inf. Theory, vol. 64, no. 9, pp. 6443--6453, Sep. 2018.
\bibitem{Carlet19} C. Carlet, On APN exponents, characterizations of differentially uniform functions by the Walsh transform, and related cyclic-difference-set-like structures. Des. Codes Cryptogr. 87 (2019), no. 2-3, 203--224.
\bibitem{CCZ} C. Carlet, P. Charpin and V. Zinoviev, Codes, Bent Functions and Permutations Suitable For DES-like Cryptosystems. Designs, Codes and Cryptography 15, 125--156 (1998).
\bibitem{CD} C. Carlet and C. Ding, Highly nonlinear mappings. J. Complexity 20 (2004), no. 2-3, 205--244.
\bibitem{CV} F. Chabaud and S. Vaudenay, Links between differential and linear cryptanalysis, in Advances in Cryptology–EUROCRYPT (Lecture Notes in Computer Science), vol. 950. Berlin, Germany: Springer, 1995, pp. 356--365.
\bibitem{EFRST} P. Ellingsen, P. Felke, C. Riera, P. Stănică and A. Tkachenko, C-Differentials, Multiplicative Uniformity, and (Almost) Perfect c-Nonlinearity, in IEEE Transactions on Information Theory, vol. 66, no. 9, pp. 5781-5789, Sept. 2020.
\bibitem{HPR+21} S. U. Hasan, M. Pal, C. Riera, and P. St\v{a}nic\v{a}, \emph{On the $c$-differential uniformity of certain maps over finite fields}, Des. Codes Cryptogr., Vol. 89, issue 2, pp.221-239, 2021.
\bibitem{HPS} S. U. Hasan, M. Pal and P. Stănică, The c-Differential Uniformity and Boomerang Uniformity of Two Classes of Permutation Polynomials, in IEEE Transactions on Information Theory, vol. 68, no. 1, pp. 679-691, Jan. 2022.
\bibitem{JKK22} J. Jeong, N. Koo, and S. Kwon, \emph{Investigations of c-Differential Uniformity of Permutations with Carlitz Rank 3}, arXiv:2202.02185.
\bibitem{LRS} C. Li, C. Riera and P. Stănică, Low c-differentially uniform functions via an extension of Dillon's switching method, arXiv:2204.08760
\bibitem{MMM}S. Mesnager, B. Mandal and M. Msahli, Survey on recent trends towards generalized differential and boomerang uniformities. Cryptogr. Commun. 14 (2022), no. 4, 691--735.
\bibitem{MRSYZ}S. Mesnager, C. Riera, P. Stănică, H. Yan and Z. Zhou, "Investigations on c-(Almost) Perfect Nonlinear Functions," in IEEE Transactions on Information Theory, vol. 67, no. 10, pp. 6916-6925, Oct. 2021.
\bibitem{Nyberg} K. Nyberg, Differentially uniform mappings for cryptography. Advances in cryptology—EUROCRYPT '93 (Lofthus, 1993), 55--64, Lecture Notes in Comput. Sci., 765, Springer, Berlin, 1994.
\bibitem{WZ} X. Wang, D. Zheng, Several classes of PcN power functions over finite fields, Discrete Applied Mathematics. 322 (2022), 171-182.
\bibitem{WLZ} Y. Wu, N. Li, X. Zeng, New PcN and APcN functions over finite fields, Designs Codes Crypt. 89 (2021), 2637–2651.
\bibitem{Y}H. Yan, On -1-differential uniformity of ternary APN power functions, Cryptogr. Commun. 2 (2022), 357–369.
\bibitem{ZH}Z. Zha, L. Hu, Some classes of power functions with low c-differential uniformity over finite fields, Des. Codes Cryptogr. 89 (2021) 1193--1210.
\end{thebibliography}
\end{document}